\documentclass[conference]{IEEEtran}

\IEEEoverridecommandlockouts

\usepackage{graphicx}
\usepackage[colorlinks,bookmarksopen,bookmarksnumbered,citecolor=blue,urlcolor=red]{hyperref}
\usepackage{amssymb}
\usepackage{amsbsy}
\usepackage{ulem}
\usepackage{upgreek}
\usepackage{algorithm}
\usepackage{algorithmicx}
\usepackage{algpseudocode}
\usepackage{threeparttable}
\usepackage[cmex10]{amsmath}
\usepackage{array}
\usepackage[tight,footnotesize]{subfigure}
\usepackage{caption}
\usepackage{url}
\usepackage{epstopdf}
\usepackage{verbatim}
\usepackage{color}
\usepackage{wrapfig}
\graphicspath{{figures/}}
\usepackage{mathtools}

\DeclarePairedDelimiter\ceil{\lceil}{\rceil}
\DeclarePairedDelimiter\floor{\lfloor}{\rfloor}

\algrenewcommand\textproc{}

\begin{document}

\newtheorem{thm}{Theorem}
\newtheorem{cor}{Corollary}
\newtheorem{cl}{Claim}
\newtheorem{lma}{Lemma}
\newtheorem{defi}{Definition}
\newtheorem{proper}{Property}
\newcommand{\Mod}[1]{\ (\mathrm{mod}\ #1)}
\newcommand{\etal}{\textit{et al.}}

\newcommand{\myfont}{\fontsize{23pt}{20pt}\selectfont}

\title{\myfont Towards An Optimal Solution to Place Bistatic Radars\\ \vspace{5.7pt} for Belt Barrier Coverage with Minimum Cost\vspace{-0.3cm}}

\author{\IEEEauthorblockN{Tu N. Nguyen\IEEEauthorrefmark{1}, Bing-Hong Liu\IEEEauthorrefmark{2}, My T. Thai\IEEEauthorrefmark{3}, and Ivan Djordjevic\IEEEauthorrefmark{4}\\}

\IEEEauthorblockA{\IEEEauthorrefmark{1}Department of Computer Science, Kennesaw State University, Marietta, GA 30060, USA\\}
\IEEEauthorblockA{\IEEEauthorrefmark{2}Department of Electronic Engineering, National Kaohsiung University of Science and Technology, Kaohsiung, Taiwan\\}
\IEEEauthorblockA{\IEEEauthorrefmark{3}
Department of Computer and Information Science and Engineering, University of Florida, Gainesville, FL 32612, USA\\}
\IEEEauthorblockA{\IEEEauthorrefmark{4}
Department of Electrical and Computer Engineering, University of Arizona, Tucson, AZ 85721, USA\\}

\it{Email: tu.nguyen@kennesaw.edu, bhliu@nkust.edu.tw, mythai@cise.ufl.edu, ivan@arizona.edu}\vspace{-0.5cm}}

\IEEEcompsoctitleabstractindextext{

\begin{abstract}
With the rapid growth of threats, sophistication and diversity in the manner of intrusion, traditional belt barrier systems are now faced with a major challenge of providing high and concrete coverage quality to expand the guarding service market. Recent efforts aim at constructing a belt barrier by deploying bistatic radar(s) on a specific line regardless of the limitation on deployment locations, to keep the width of the barrier from going below a specific threshold and the total bistatic radar placement cost is minimized, referred to as the Minimum Cost Linear Placement (MCLP) problem.
The existing solutions are \textit{heuristic} and their validity is tightly bound by the \textit{barrier width} parameter that these solutions only work for a fixed barrier width value. In this work, we propose an optimal solution, referred to as the Opt$\_$MCLP, for the ``open MCLP problem" that works for full range of the barrier width. Through rigorous theoretical analysis and experimentation, we demonstrate that the proposed algorithms perform well in terms of placement cost reduction and barrier coverage guarantee.
\end{abstract}
}

\noindent \maketitle

\IEEEdisplaynotcompsoctitleabstractindextext

\IEEEpeerreviewmaketitle

\section{Introduction}
Images of moats 
remind us of the system of guarding ancient castles and the idea is still used today, yet at a higher level. With the development of sensing and communication technology, instead of using moat, barriers of sensors and radars, referred to as the bistatic radar systems, have been building to guard not only critical places but also spaces and national borders. Many applications (e.g., typically for military to detect targets and intruders) use a bistatic radar system that comprises at least a radar signal transmitters and at a radar signal least receiver\footnote{For ease of exposition, hereafter referred to as transmitter(s) and the receiver(s), respectively.}, wherein the transmitter(s) and the receiver(s) are located in a different location, to form the \textit{barrier coverage} \cite{7249400,6416889,8031350,KARATAS2018129,8247236,CHEN2015129}.

In recent years, there are efforts in the existing literature to design barrier coverage using
radar \cite{7087387,7087389,7944669,TIAN2017222,9121308,9404350}, wherein a radar uses radio waves to detect an object by producing the radio waves and collecting the echo signal reflected off from the target,
giving information about the object's location and speed.
Emerging research problems in recent years are how to enhance the quality of barrier coverage and how to efficiently deploy bistatic radar systems while meeting
quality requirements.
In terms of barrier coverage quality, one of the important aspects that reflects the quality of bistatic radar systems is the width of the barrier coverage area.
Recent efforts aim at constructing a linear belt barrier with pre-defined width by deploying bistatic radar(s) on a specific line
such that the total bistatic radar placement cost is minimized, referred to as the Minimum Cost Linear Placement (MCLP) problem \cite{7087387}.
The sensing model of bistatic radars (Cassini oval sensing model \cite{7087387}) is in fact very complex and
the shape of sensing areas are, therefore, varied according to the variation of distance between the transmitter(s) and the receiver(s) in the bistatic radar system.
The validity of the most recent solutions \cite{7087387,7087389} is unfortunately tightly bound by the ``barrier width" parameter
that these solutions only work for a pre-fixed barrier width value. Thus, they fail to solve the MCLP problem with flexible width ranges of barrier.

Specifically, the most recent solutions proposed for the MCLP problem \cite{7087387,7087389} can only work for covering a belt-shaped area
having width less than or equal to $\frac{2\zeta_{max}}{\sqrt{3}}$, namely $0 \leq 2\omega \leq \frac{2\zeta_{max}}{\sqrt{3}}$, where $2\omega$ denotes the width of the belt-shaped area and $\zeta_{max}$ is the radius of the coverage circle centered at the location of transmitter and receiver
when the transmitter and the receiver are located at the same location (we will present how to obtain $\zeta_{max}$ in detail in $\S$\ref{subsection:system_model_and_the_problem}).
In other words, the barrier built by the bistatic radar system using their solution cannot cover the belt-shaped with width greater than $\frac{2\zeta_{max}}{\sqrt{3}}$, namely $\frac{2\zeta_{max}}{\sqrt{3}} < 2\omega < 2\zeta_{max}$\footnote{$\zeta_{max}$ is the radius of the coverage circle centered at the location of transmitter and receiver
when the transmitter and the receiver are located at the same location, therefore, the maximum width of the belt-shaped area in the MCLP problem has to be less than $2\zeta_{max}$}.
In this work, we seek to design an optimal algorithm for the ``open MCLP problem" that works for full range of the barrier width,
to achieve maximum coverage for the barrier.
In addition, the existing solutions \cite{7087387} proposed for the MCLP problem are still ``heuristic" (not an optimal solution)
because boundary conditions are not considered.
The main contributions of this paper are summarized as follows.

\begin{itemize}

  \item We investigate the problem in both cases when $0 \leq 2\omega \leq \frac{2\zeta_{max}}{\sqrt{3}}$
  and $\frac{2\zeta_{max}}{\sqrt{3}} < 2\omega <  2\zeta_{max}$.
  We propose an optimal algorithm -- dubbed the Opt$\_$MCLP -- to find the optimal solution for the ``open MCLP problem" that works for full range of the barrier width ($0 \leq 2\omega <  2\zeta_{max}$).

  \item We provide rigorous theoretical analyses to demonstrate the correctness of the proposed optimal solution.

  \item Extensive simulations are conducted to demonstrate the performance of the Opt$\_$MCLP for the MCLP problem.
  The obtained results show that the Opt$\_$MCLP provides a significantly higher performance than the existing method.
\end{itemize}

\textbf{Organization:}
The remaining sections of this paper are organized as follows.
The remaining sections of this paper are
organized as follows. The basic mathematical notations and
system model are initially introduced, and simple examples
are used to present the key ideas behind the proposed work in $\S$\ref{section:preliminaries}.
An optimal solution, termed the Opt$\_$MCLP is proposed for the full range of the MCLP problem in $\S$\ref{section:opti_linear_second_range}.
Simulations are evaluated in Section \ref{section:Simulation}, and
the paper is concluded in Section \ref{section:Conclusion}.

\vspace{-15pt}
\section{Preliminaries}\label{section:preliminaries}

\subsection{System Model and Problem Definition}\label{subsection:system_model_and_the_problem}

A bistatic radar is composed of at least a transmitter and a receiver that are separated and often located at different positions.
In a bistatic radar, the transmitter $t$ is responsible for
producing and propagating the radio waves and the receiver $r$ can detect an object (target) $x$
using the echo signal reflected off from $x$ if the received signal-to-noise ratio (SNR) is not less than a threshold $\gamma$.
For any target $x$ and a bistatic radar paired by transmitter $t$ and receiver $r$,
the SNR of the radar signal that is sent from $t$, reflected by $x$, and received by $r$ can be obtained \cite{willis2005bistatic} as follows:
\begin{equation}\label{eq:SNR}
SNR(t,x,r) = \frac{K}{d^2(t,x) \cdot d^2(x,r)},
\end{equation}
where $d(t,x)$ (or $d(x,r)$) denotes the Euclidean distance between $t$ (or $r$) and $x$; and
$K$ represents a constant that is determined by a bistatic radar's physical characteristics, such as
the antenna's power gain and the transmission power.
When the minimum threshold $\gamma$ is given, for any pair of transmitter $t$ and receiver $r$,
the possible locations of targets $x$ with $SNR(t,x,r)=\gamma$ can be characterized
by the locus of points such that the product of the distances to $t$ and $r$, namely $d(x,t) \cdot d(x,r)$, is the constant equal to
$\zeta^2_{max}$, where $\zeta_{max}$ is a constant and denotes $\sqrt[4]{\frac{K}{\gamma}}$.
The locus of points $x$, which will be a \textit{closed curve} or \textit{a pair of closed curves},
is known as a Cassini oval \cite{Cassinioval} as depicted in Fig. \ref{fig_radar_region_shape}.
For any target $y$ within the Cassini oval, the product of the distances from $y$ to $t$ and $r$
is not greater than $\sqrt{\frac{K}{\gamma}}$, that is, $d(y,t) \cdot d(y,r) \leq \zeta^2_{max}$; and therefore, we have that $SNR(t,y,r) \geq \gamma$ and $y$, therefore, can be detected by $r$.
Hereafter, a point $z$ in the plane is said to be covered by a bistatic radar paired by transmitter $t$ and receiver $r$
if $z$ is within the area surrounded by the Cassini oval with focal points at $t$ and $r$.

When $\zeta_{max}$ is given, the shape of the Cassini oval with focal points at transmitter $t$ and receiver $r$ is determined by the distance between $t$ and $r$, that is, $d(t,r)$. Four shape types of the  Cassini  oval depicted in Fig. \ref{fig_radar_region_shape} are listed as follows with different range of $d(t,r)$ \cite{7087387}:
\begin{itemize}
  \item an ellipse curve (Fig. \ref{fig_shape_01}) if $0 \leq d(t,r) < \sqrt{2}\zeta_{max}$;
  \item a waist curve (Fig. \ref{fig_shape_02}) if $\sqrt{2}\zeta_{max} \leq d(t,r) < 2\zeta_{max}$;
  \item a lemniscate of Bernoulli (Fig. \ref{fig_shape_03}) if $d(t,r) = 2\zeta_{max}$;
  \item a pair of closed curves (Fig. \ref{fig_shape_04}) if $d(t,r) > 2\zeta_{max}$.
\end{itemize}

\begin{figure}
\center
\subfigure[]{\includegraphics[width=3cm]{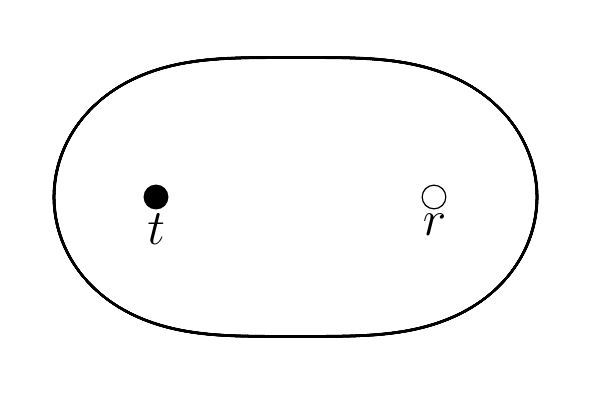}\label{fig_shape_01}}
\subfigure[]{\includegraphics[width=3cm]{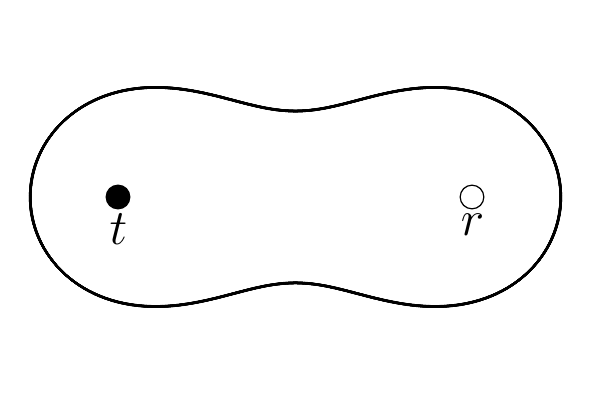}\label{fig_shape_02}}
\subfigure[]{\includegraphics[width=3cm]{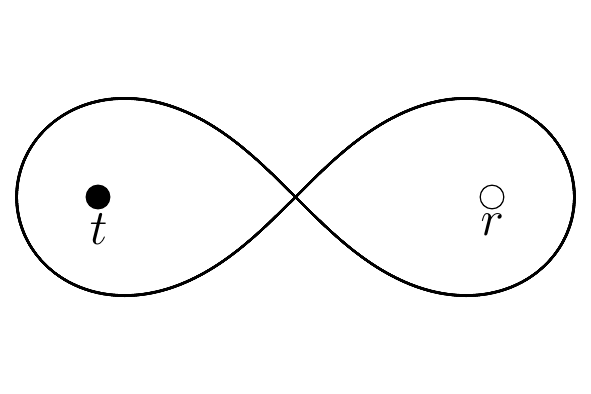}\label{fig_shape_03}}
\subfigure[]{\includegraphics[width=3cm]{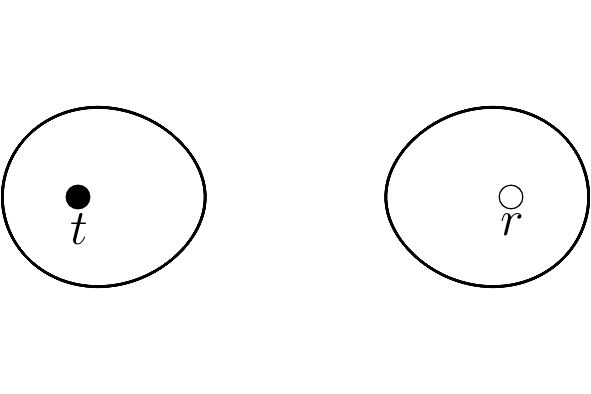}\label{fig_shape_04}}
\caption{Four shape types of the Cassini oval with focal points at transmitter $t$ and receiver $r$, where small solid and hollow circles denote transmitters and receivers, respectively. The shape types with $0 \leq d(t,r) < \sqrt{2}\zeta_{max}$,
$\sqrt{2}\zeta_{max} \leq d(t,r) < 2\zeta_{max}$, $d(t,r) = 2\zeta_{max}$, and $d(t,r) > 2\zeta_{max}$
are shown in (a), (b), (c), and (d), respectively.}
\vspace{-15pt}
\label{fig_radar_region_shape}
\end{figure}

When multiple transmitters and receivers are deployed in an area,
transmitters and receivers can be paired to form multiple bistatic radars.
Here, transmitters are assumed to operate with orthogonal frequency such that
mutual interferences at a receiver can be avoided \cite{5703085}.
Therefore, each receiver can be paired with different transmitters to form bistatic radars by switching the frequency.
In addition, because receivers can receive the radar signal sent from a transmitter,
multiple receivers can also be paired with the same transmitter.
Because multiple bistatic radars are often used to detect targets,
an area is said to be covered by the set of transmitters $T$ and the set of receivers $R$
hereafter, if for any point $z$ in the area, $z$ is covered by at least one bistatic radar paired by a transmitter $t \in T$ and a receiver $r \in R$.

Let $c_t$ and $c_r$ be the placement/deployment costs of a transmitter and a receiver, respectively.
Also let $Area$ be a belt-shaped (rectangle) area with length $L$ and width $W$, where
$L \geq W > 0$.
In addition, when a transmitter and a receiver are located at the same location,
the shape of the generated Cassini oval will be a circle centered at the transmitter/receiver with radius $\zeta_{max}$ \cite{Cassinioval}.
While we are given $\zeta_{max}$, $c_t$, $c_r$, and $Area$ with width less than $2\zeta_{max}$,
the Minimum Cost Linear Placement (MCLP) problem
is to deploy a set of transmitters $T$ and a set of receivers $R$ on the \textit{line} that goes through
the middle points of the shorter sides of the $Area$, such that the $Area$ is fully covered and the total placement cost of bistatic radar(s), namely $c_t \times |T| + c_r \times |R|$,
is minimized,
where $|T|$ and $|R|$ denote the cardinalities of $T$ and $R$, respectively.

\begin{figure}[t]
\center
\subfigure{\includegraphics[width=8cm]{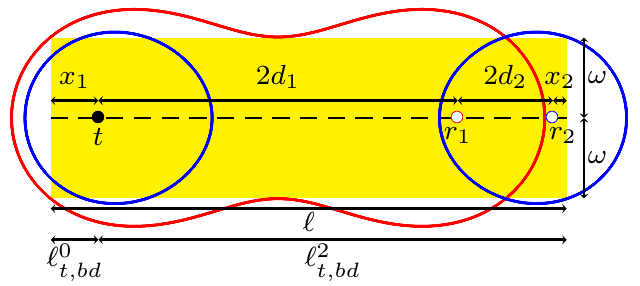}\label{fig_example_maxlength_01}}
\caption{Examples of covering a rectangle area by using one transmitter and two receivers (best viewed in color)}.
\vspace{-15pt}
\label{fig_maxlength}
\end{figure}

For the MCLP problem, let $2\omega$ be the width of the $Area$, and $\beta$ denote $\frac{\zeta_{max}}{\omega}$.
Although solutions for the MCLP problem are proposed in \cite{7087387},
the solutions only work for the case of $\beta \geq \sqrt{3}$, that is,
$\frac{\zeta_{max}}{\omega} \geq \sqrt{3}$ ($2\omega \leq \frac{2\zeta_{max}}{\sqrt{3}}$).
In addition, the width of the $Area$ in the MCLP problem is always less than $2\zeta_{max}$, which implies
that these solutions are not valid for the case of $\frac{2\zeta_{max}}{\sqrt{3}} < 2\omega <  2\zeta_{max}$ ($1 < \beta < \sqrt{3}$).
This motivates us to explore an optimal solution for the MCLP problem that especially works for \textit{full
range of the coverage $Area$ width} ($0 \leq 2\omega <  2\zeta_{max}$).

\section{The Optimal Solution for the Minimum Cost Linear Placement (MCLP) Problem}\label{section:opti_linear_second_range}

By Fig. \ref{fig_radar_region_shape},
we have that when transmitter $t$ and receiver $r$ are close enough, the area covered by $t$ and $r$ will be an ellipse or waist shape.
In order to cover a rectangle area with width $2\omega$,
the distance between the upper and the lower parts of the ellipse (or the waist) curve has to be not less than $2\omega$.
Take Fig. \ref{fig:example_line_deployment_with_larger_width_with_type_shapes}, for example, where $\omega = \frac{\zeta_{max}}{\sqrt{3}}$.
When the covered area is a waist shape as shown in Fig. \ref{fig_secondrange_01},
because the upper and the lower parts of the curve generated by transmitter $t$ and receiver $r$ are symmetric,
and $d(E,F) = d(t,D) = d(r,C) = \frac{\zeta_{max}}{\sqrt{3}}$, $t$ and $r$ can cover a rectangle area
with width $2\omega$. Similar example with an ellipse shape is shown in Fig. \ref{fig_secondrange_02}.
This motivates us to find an optimal deployment of a transmitter and a receiver such that the width of the covered rectangle can be satisfied
and its length can be maximized.
As the examples in Fig. \ref{fig:example_line_deployment_with_larger_width_with_type_shapes}, when a rectangle is fully covered by transmitter $t$ and receiver $r$,
$t$ (or $r$) may be at a distance of $\theta$ from the closest vertical boundary of the covered rectangle.
By the observation,
we show that transmitter $t$ and receiver $r$ with $d(t,r) = \sqrt{\beta^2 \zeta^2_{max} - \omega^2}$
can cover a rectangle with width $2\omega$ and length $\sqrt{\beta^2 \zeta^2_{max} - \omega^2}$;
and that transmitter $t$ and receiver $r$ with $d(t,r) = \omega \sqrt{\frac{\beta^4}{\frac{\theta'^2}{\omega^2} + 1} - 1} - \theta'$
can cover a rectangle with width $2\omega$ and maximum length $\omega \sqrt{\frac{\beta^4}{\frac{\theta'^2}{\omega^2} + 1} - 1} + \theta'$
by Lemmas \ref{lma:second_range:max_dtr_with_thetas}-\ref{lma:second_range:rectangle_is_covered},
where $\theta'$ denotes the $\theta$ in $[0, \sqrt{\beta^2 \zeta^2_{max} - \omega^2}]$
having maximum value of $\omega \sqrt{\frac{\beta^4}{\frac{\theta^2}{\omega^2} + 1} - 1} + \theta$
and can be obtained by Newton's method \cite{atkinson2008introduction}.
In addition, Lemma \ref{lma:second_range:max_length_foroneTx} shows that the length of the rectangle with width $2\omega$
covered by one transmitter and any number of receivers is at most $2\sqrt{\beta^2 \zeta^2_{max} - \omega^2}$.
Note that we can use one transmitter $t$ and two receivers, $r_1$ and $r_2$, by the sequence $(r_1, t, r_2)$,
to cover a rectangle with width $2\omega$ and length $2\sqrt{\beta^2 \zeta^2_{max} - \omega^2}$ by Lemma \ref{lma:second_range:rectangle_is_covered},
as the example in Fig. \ref{fig:example_line_deployment_with_larger_width}.

\begin{figure}[t]
\center
\subfigure[]{\includegraphics[width=4cm]{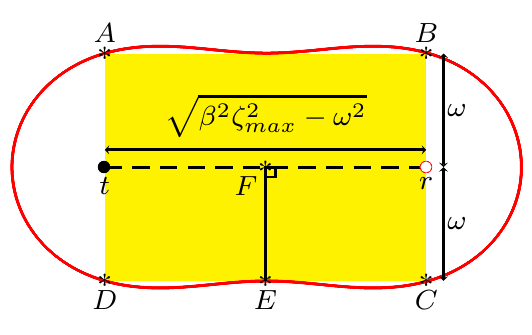}\label{fig_secondrange_01}}
\subfigure[]{\includegraphics[width=3.7cm]{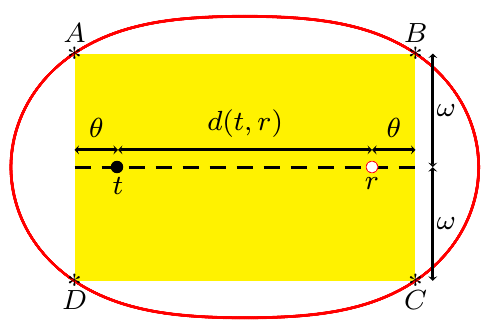}\label{fig_secondrange_02}}
\caption{Examples of deploying transmitter $t$ and receiver $r$, where $\omega = \frac{\zeta_{max}}{\sqrt{3}}$. The distances between $r$ and the boundary in (a) and (b) are 0 and $\theta$, respectively.}
\vspace{-15pt}
\label{fig:example_line_deployment_with_larger_width_with_type_shapes}
\end{figure}

\begin{lma}\label{lma:second_range:max_dtr_with_thetas}
When $1 < \beta < \sqrt{3}$ and a rectangle with width $2\omega$ is covered by transmitter $t$ and receiver $r$,
$d(t,r) + \theta$ is maximized if and only if $d(t,r) = \sqrt{\beta^2 \zeta^2_{max} - \omega^2}$ and $\theta = 0$, where
$\theta$ denotes the distance between $t$ (or $r$) and the closest vertical boundary of the rectangle.
In addition, $d(t,r) + 2 \theta$ is maximized if and only if
$d(t,r) = \omega \sqrt{\frac{\beta^4}{\frac{\theta'^2}{\omega^2} + 1} - 1} - \theta'$,
where $\theta'$ denotes the $\theta$ in $[0, \sqrt{\beta^2 \zeta^2_{max} - \omega^2}]$
having maximum value of $\omega \sqrt{\frac{\beta^4}{\frac{\theta^2}{\omega^2} + 1} - 1} + \theta$.
\end{lma}
\begin{proof}
Without loss of generality, let $\Box ABCD$ be the rectangle covered by $t$ and $r$,
as the rectangle in Fig. \ref{fig_secondrange_01} or Fig. \ref{fig_secondrange_02}.
Because the upper and the lower parts (or, the left and the right parts) of the curve generated by
$t$ and $r$ are symmetric, the proof suffices
to show that the cases hold when the point $B$ is covered.
When the point $B$ is exactly covered by $t$ and $r$,
we have that $d(t,B) \cdot d(r,B) = \zeta^2_{max}$, which implies that
$\sqrt{ (d(t,r) + \theta)^2 + \omega^2 } \cdot \sqrt{\theta^2 + \omega^2} = \zeta^2_{max}$.
We thus have that $d(t,r) = \omega \sqrt{\frac{\beta^4}{\frac{\theta^2}{\omega^2} + 1} - 1} - \theta$.
Therefore, $d(t,r) + \theta = \omega \sqrt{\frac{\beta^4}{\frac{\theta^2}{\omega^2} + 1} - 1}$.
Because $\omega$ and $\beta$ are constants, $d(t,r) + \theta$ is maximized if and only if $\theta = 0$.
When $\theta = 0$, $d(t,r) + \theta = d(t,r) = \omega \sqrt{\beta^4 - 1} = \sqrt{\beta^2 \zeta^2_{max} - \omega^2}$.
In addition,
$d(t,r) + 2\theta = \omega \sqrt{\frac{\beta^4}{\frac{\theta^2}{\omega^2} + 1} - 1} + \theta$.
It is clear that $d(t,r) + 2\theta$ is maximized if and only if $\theta \geq 0$.
Assume that $\theta = \sqrt{\beta^2 \zeta^2_{max} - \omega^2} + \epsilon$, where $\epsilon > 0$.
Because $d(t,r) \geq 0$ and $1 < \beta < \sqrt{3}$, we have that $d(t,B) \cdot d(r,B) = \sqrt{ (d(t,r) + \theta)^2 + \omega^2 } \cdot \sqrt{\theta^2 + \omega^2}$
$\geq (\sqrt{ \theta^2 + \omega^2 })^2$
$= (\sqrt{\beta^2 \zeta^2_{max} - \omega^2} + \epsilon)^2 + \omega^2$
$= \beta^2 \zeta^2_{max} + 2 \epsilon \sqrt{\beta^2 \zeta^2_{max} - \omega^2} + \epsilon^2$
$> \zeta^2_{max}$.
This constitutes a contradiction because the point $B$ has to be covered, that is, $d(t,B) \cdot d(r,B) \leq \zeta^2_{max}$.
Therefore, we have that $\epsilon \leq 0$, that is, $\theta \leq \sqrt{\beta^2 \zeta^2_{max} - \omega^2}$.
Because $\theta \geq 0$, we have that $d(t,r) + 2\theta$ is maximized if and only if $0 \leq \theta \leq \sqrt{\beta^2 \zeta^2_{max} - \omega^2}$.
Let $\theta'$ be the $\theta$ in $[0, \sqrt{\beta^2 \zeta^2_{max} - \omega^2}]$
having maximum value of $\omega \sqrt{\frac{\beta^4}{\frac{\theta^2}{\omega^2} + 1} - 1} + \theta$.
We have that $d(t,r) + 2\theta$ is maximized
if and only if $d(t,r) + 2\theta' = \omega \sqrt{\frac{\beta^4}{\frac{\theta'^2}{\omega^2} + 1} - 1} + \theta'$,
that is, $d(t,r) = \omega \sqrt{\frac{\beta^4}{\frac{\theta'^2}{\omega^2} + 1} - 1} - \theta'$,
which completes the proof.
\end{proof}

\begin{lma}\label{lma:second_range:rectangle_is_covered}
When $1 < \beta < \sqrt{3}$,
the transmitter $t$ and the receiver $r$ that are at a distance of $\sqrt{\beta^2 \zeta^2_{max} - \omega^2}$ apart can
cover a rectangle with width $2\omega$ and length equal to $\sqrt{\beta^2 \zeta^2_{max} - \omega^2}$.
In addition, the transmitter $t$ and the receiver $r$
that are at a distance of $d(t,r) = \omega \sqrt{\frac{\beta^4}{\frac{\theta'^2}{\omega^2} + 1} - 1} - \theta'$ apart
can cover a rectangle with width $2\omega$ and length equal to $d(t,r) + 2\theta'$,
where $\theta'$ is defined in Lemma \ref{lma:second_range:max_dtr_with_thetas}.
\end{lma}
\begin{proof}
The proof has to show that S1) if $d(t,r) = \sqrt{\beta^2 \zeta^2_{max} - \omega^2}$,
$t$ and $r$ can fully cover a rectangle with width $2\omega$ and length $\sqrt{\beta^2 \zeta^2_{max} - \omega^2}$,
and that S2) if $d(t,r) = \omega \sqrt{\frac{\beta^4}{\frac{\theta'^2}{\omega^2} + 1} - 1} - \theta'$,
$t$ and $r$ can fully cover a rectangle with width $2\omega$ and length $\omega \sqrt{\frac{\beta^4}{\frac{\theta'^2}{\omega^2} + 1} - 1} + \theta'$,
where $\theta'$ denotes the $\theta$ in $[0, \sqrt{\beta^2 \zeta^2_{max} - \omega^2}]$
having maximum value of $\omega \sqrt{\frac{\beta^4}{\frac{\theta^2}{\omega^2} + 1} - 1} + \theta$.
The proof of S2 is omitted here due to the similarity.

For S1, because $\frac{\zeta_{max}}{\sqrt{3}} < \omega <  \zeta_{max}$, we have that
$d(t,r) = \sqrt{\beta^2 \zeta^2_{max} - \omega^2} = \sqrt{\frac{\zeta^4_{max} - \omega^4}{\omega^2}}$
$< \sqrt{\frac{\zeta^4_{max} - (\frac{\zeta_{max}}{\sqrt{3}})^4}{(\frac{\zeta_{max}}{\sqrt{3}})^2}}$
$= \sqrt{\frac{8 \zeta^2_{max}}{3}} < 2 \zeta_{max}$, which implies that the area covered by $t$ and $r$ will be an ellipse or waist shape
by the results in Fig. \ref{fig_radar_region_shape} \cite{7087387}.
Let $\Box ABCD$ be a rectangle with width $2\omega$ and length $\sqrt{\beta^2 \zeta^2_{max} - \omega^2}$, as the rectangle in Fig. \ref{fig_secondrange_01}.
Due to the fact that the upper and the lower parts of the curve generated by $t$ and $r$ are symmetric,
the proof of S1 suffices to show that $\Box trCD$ is within the area covered by $t$ and $r$.
Let $\overrightarrow{tD'}$ (or $\overrightarrow{rC'}$) be the line from $t$ (or, from $r$),
perpendicular to $\overline{tr}$, and intersected by the lower part of the generated curve at point $D'$ (or $C'$).
Also let the curve $\widetilde{D'C'}$ be the lower part of the curve generated by $t$ and $r$ with endpoints $D'$ and $C'$.
Because the curve generated by $t$ and $r$ is an ellipse or waist shape, and
the left and the right parts of the generated curve are symmetric,
we have that the minimum distance between $\overline{tr}$ and $\widetilde{D'C'}$ is the minimum value of $d(D',\overline{tr})$, $d(C',\overline{tr})$, and $d(E',\overline{tr})$, where $d(D',\overline{tr})$, $d(C',\overline{tr})$, and $d(E',\overline{tr})$ denote the minimum distance from $D'$, $C'$, and $E'$, respectively, to $\overline{tr}$, and $E'$ denotes the midpoint of $\widetilde{D'C'}$.
We thus also have that $\Box trCD$ is fully covered if $d(D',\overline{tr}) \geq d(D,\overline{tr}) = \omega$, $d(C',\overline{tr}) \geq d(C,\overline{tr})  = \omega$, and $d(E',\overline{tr}) \geq d(E,\overline{tr}) = \omega$, that is, the points $C$, $D$, and $E$ have to be covered by $t$ and $r$, where the point $E$ is the midpoint of $\overline{CD}$.
For the point $D$, we have that $d(t,D) \cdot d(r,D) = \omega \cdot \sqrt{d^2(r,t)+d^2(t,D)}$
$= \zeta^2_{max}$, implying that $D$ is covered by $t$ and $r$.
Similarly, the point $C$ is also covered by $t$ and $r$.
For the point $E$, we assume that $E$ is not covered by $t$ and $r$.
That is, $d(t,E) \cdot d(r,E) = (\sqrt{(\frac{d(t,r)}{2})^2+\omega^2})^2 > \zeta^2_{max}$, and thus,
we have that $3 \omega^4 - 4 \omega^2 \zeta^2_{max} + \zeta^4_{max} > 0$, implying that $(3 \omega^2 - \zeta^2_{max})(\omega^2 - \zeta^2_{max})> 0$.
We have that $\omega < \frac{\zeta_{max}}{\sqrt{3}}$ or $\omega > \zeta_{max}$, which constitutes a contradiction
because $\frac{\zeta_{max}}{\sqrt{3}} < \omega <  \zeta_{max}$. Therefore,
the points $C$, $D$, and $E$ are covered by $t$ and $r$, and $\Box ABCD$ is fully covered by $t$ and $r$.
This completes the proof of S1, and thus, the proof of the lemma is also completed.
\end{proof}

\begin{lma}\label{lma:second_range:max_length_foroneTx}
When $1 < \beta < \sqrt{3}$ and a transmitter $t$ is given, the maximum length of the rectangle with width $2\omega$
covered by $t$ and any number of receivers is equal to $2\sqrt{\beta^2 \zeta^2_{max} - \omega^2}$.
\end{lma}
\begin{proof}
Assume that there exists a sequence $\Psi = (R^{n_{p_0}}, t, R^{n_{p_1}})$ deployed on a line such that
the length of the rectangle covered by $\Psi$, denoted by $\ell(\Psi)$, is greater than $2\sqrt{\beta^2 \zeta^2_{max} - \omega^2}$, where
$n_{p_0}, n_{p_1} \geq 0$ and $R^{n_{p_i}}$ denotes the set of $n_{p_i}$ receivers.
By Lemmas \ref{lma:second_range:max_dtr_with_thetas}-\ref{lma:second_range:rectangle_is_covered},
we have that the maximum length of the covered rectangle from $t$ to one side boundary is at most $\sqrt{\beta^2 \zeta^2_{max} - \omega^2}$.
This implies that $2 \sqrt{\beta^2 \zeta^2_{max} - \omega^2} \geq \ell(\Psi) > 2\sqrt{\beta^2 \zeta^2_{max} - \omega^2}$,
which constitutes a contradiction. This thus completes the proof.
\end{proof}

\begin{figure}
\center
{\includegraphics[width=7cm]{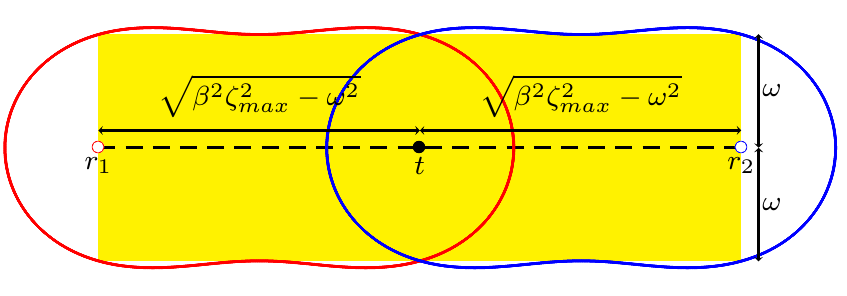}}
\caption{Example of deploying one transmitter with two receivers when $\omega = \frac{\zeta_{max}}{\sqrt{3}}$.}
\label{fig:example_line_deployment_with_larger_width}
\end{figure}

Let $I^m_n$ be the sequence $(r_1, t_1, r_2, t_2, \ldots, r_m, t_m)$ if $n = m$, and be
the sequence $(r_1, t_1, r_2, t_2, \ldots, r_m, t_m, r_{m+1})$ if $n = m+1$.
The proposed placement for $1 < \beta < \sqrt{3}$, termed the Rotated Placement hereafter,
is to place $m$ transmitters and $n$ receivers following the sequence $I^m_n$ sequentially,
where the distance between a transmitter and its adjacent receiver is $\sqrt{\beta^2 \zeta^2_{max} - \omega^2}$
by Lemma \ref{lma:second_range:rectangle_is_covered},
as the examples in Fig. \ref{fig_line_deployment_sequence_second_range}.

\begin{figure}
\center
{\includegraphics[width=8cm]{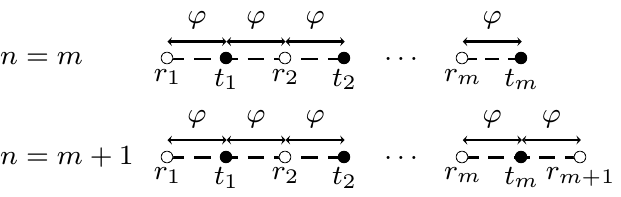}}
\caption{Examples of deploying $m$ transmitters and $n$ receivers by the Rotated Placement, where $\varphi$ denotes $\sqrt{\beta^2 \zeta^2_{max} - \omega^2}$.}
\vspace{-15pt}
\label{fig_line_deployment_sequence_second_range}
\end{figure}

The procedure MinCost$\_$RP is designed to find a sequence of transmitters and receivers with minimum cost to
cover a rectangle area $Area$ with length $\ell$ and width $2\omega$.
By Lemma \ref{lma:second_range:max_dtr_with_thetas}-\ref{lma:second_range:rectangle_is_covered},
if $\ell \leq \omega \sqrt{\frac{\beta^4}{\frac{\theta'^2}{\omega^2} + 1} - 1} + \theta'$,
the $Area$ can be covered by one transmitter and one receiver.
Otherwise, by Lemma \ref{lma:second_range:max_dtr_with_thetas}-\ref{lma:second_range:max_length_foroneTx},
one transmitter coupled with two receivers can be used to cover a rectangle with length $2\sqrt{\beta^2 \zeta^2_{max} - \omega^2}$ that is called a unit hereafter.
The procedure MinCost$\_$RP is then to see at least how many units are required.
The necessary numbers of transmitters and receivers are calculated and stored in $numTx$ and $numRx$, respectively.
When the $Area$ is not fully covered,
additional one transmitter is added if the length of the remaining uncovered part of the $Area$ is not greater than the half length of a unit; and otherwise,
additional one transmitter and one receiver are added.
Finally, by the obtained $numTx$ and $numRx$, the Rotated Placement is used to generate the optimal solution $I^{numTx}_{numRx}$,
which is proved in Theorem \ref{thm:mincost_RP}.

\begin{thm}\label{thm:mincost_RP}
When $1 < \beta < \sqrt{3}$, the procedure MinCost$\_$RP can find an optimal solution to
cover a rectangle area $Area$ with length $\ell$ and width $2\omega$.
\end{thm}
\begin{proof}
By Lemma \ref{lma:second_range:max_dtr_with_thetas}-\ref{lma:second_range:rectangle_is_covered},
we have that deploying one transmitter and one receiver to cover the $Area$ is an optimal solution
if $\ell \leq \omega \sqrt{\frac{\beta^4}{\frac{\theta'^2}{\omega^2} + 1} - 1} + \theta'$.
In addition,
because it is easy to verify that the difference between $numTx$ and $numRx$ in the procedure MinCost$\_$RP is at most one,
the proof suffices to show that the solution generated by the procedure MinCost$\_$RP has a minimum cost.
Three cases, including C1) $\ell = 2k\sqrt{\beta^2 \zeta^2_{max} - \omega^2}$, C2) $2k\sqrt{\beta^2 \zeta^2_{max} - \omega^2} < \ell \leq (2k+1)\sqrt{\beta^2 \zeta^2_{max} - \omega^2}$, and C3) $(2k+1)\sqrt{\beta^2 \zeta^2_{max} - \omega^2} < \ell < (2k+2)\sqrt{\beta^2 \zeta^2_{max} - \omega^2}$, are considered, where $k \geq 0$ and $\ell > \omega \sqrt{\frac{\beta^4}{\frac{\theta'^2}{\omega^2} + 1} - 1} + \theta'$. Because the proofs of C2 and C3 are similar to that of C1, the proofs of C2 and C3 are omitted.

\begin{algorithm} [http]
\label{alg:RP}
\begin{algorithmic}[1]
\Procedure{MinCost$\_$RP} {$\omega$, $\ell$}

\If {$\ell \leq \omega \sqrt{\frac{\beta^4}{\frac{\theta'^2}{\omega^2} + 1} - 1} + \theta'$, where $\theta'$ is defined in Lemma \ref{lma:second_range:max_dtr_with_thetas}}

    \State Let $I$ be the sequence $(r, t)$ with $d(t,r) = \omega \sqrt{\frac{\beta^4}{\frac{\theta'^2}{\omega^2} + 1} - 1} - \theta'$

    \State $minValue$ $\gets$ $c_t + c_r$

    \State \Return $(minValue, I)$

\EndIf

\State $numTx \gets \floor*{\frac{\ell}{2\sqrt{\beta^2 \zeta^2_{max} - \omega^2}}}$

\State $numRx \gets numTx + 1$

\State $remainingLen \gets \ell\Mod{2\sqrt{\beta^2 \zeta^2_{max} - \omega^2}}$

\If {$remainingLen > 0$}

    \State $numTx \gets numTx + 1$

    \If {$\ceil{\frac{remainingLen}{\sqrt{\beta^2 \zeta^2_{max} - \omega^2}}}  = 2$}

        \State $numRx \gets numRx + 1$

    \EndIf

\EndIf

\State Let $I^{numTx}_{numRx}$ be the sequence generated by the Rotated Placement with $numTx$ transmitters and $numRx$ receivers

\State $minValue$ $\gets$ $numTx \times c_t + numRx \times c_r$

\State \Return $(minValue, I^{numTx}_{numRx})$

\EndProcedure

\end{algorithmic}
\end{algorithm}

For C1, let $I^{m_1}_{n_1}$ and $\Psi^{m_2}_{n_2}$ be the sequences obtained by the procedure MinCost$\_$RP and the optimal solution, respectively.
Assume that $I^{m_1}_{n_1}$ is not an optimal solution.
This implies that
$\varsigma(I^{m_1}_{n_1}) > \varsigma(\Psi^{m_2}_{n_2})$, where $\varsigma(I^{m_1}_{n_1})$ (or $\varsigma(\Psi^{m_2}_{n_2})$) denotes the placement cost of $I^{m_1}_{n_1}$ (or $\Psi^{m_2}_{n_2})$, and $m_1 \neq m2$ or $n_1 \neq n_2$.
Because $\ell = 2k\sqrt{\beta^2 \zeta^2_{max} - \omega^2}$,
we have that $m_1 = k$ and $n_1 = k+1$ by the procedure MinCost$\_$RP.
By Lemma \ref{lma:second_range:max_length_foroneTx},
due to the fact that one transmitter with receivers can cover a rectangle with length at most $2\sqrt{\beta^2 \zeta^2_{max} - \omega^2}$,
at least $\floor*{\frac{2k\sqrt{\beta^2 \zeta^2_{max} - \omega^2}}{2\sqrt{\beta^2 \zeta^2_{max} - \omega^2}}} = k$ transmitters are required
to cover the $Area$.
Because $m_1 = k$, we have that C1.1) $m_2 > k$, or C1.2) $m_2 = k$ and $n_2 < n_1 = k + 1$.
Due to the fact that the proof of C1.2 is similar to that of C1.1, the proof of C1.2 is omitted here.
For C1.1, if $m_2 > k$, because $\varsigma(I^{m_1}_{n_1}) > \varsigma(\Psi^{m_2}_{n_2})$, we have that
$c_t \times k + c_r \times (k+1) > c_t \times (k+y) + c_r \times n_2$, where $m_2 = k + y$ and $y \geq 1$.
Because $y \geq 1$ and $c_t \geq c_r$, we have that $n_2 < (k+1) - \frac{yc_t}{c_r} \leq k$.
Due to the fact that the covered area is the same when receivers and transmitters are swapped,
by Lemma \ref{lma:second_range:max_length_foroneTx},
we have that one receiver with transmitters can cover a rectangle with length at most $2\sqrt{\beta^2 \zeta^2_{max} - \omega^2}$.
In addition, because $n_2 < k$, that is, $n_2 \leq k-1$, the rectangle with length at most $2(k-1)\sqrt{\beta^2 \zeta^2_{max} - \omega^2}$ can be covered by
$n_2$ receivers with transmitters, which implies that the $Area$ cannot be covered by $\Psi^{m_2}_{n_2}$.
This constitutes a contradiction, and completes the proof of C1.1.
This thus completes the proof of the theorem.
\end{proof}

For an $Area$ with width $W = 2\omega$ and length $L = \ell$, the algorithm Opt$\_$MCLP is designed as shown in Algorithm \ref{alg:opt} to combine the function ComputeMinCost, which is proposed in \cite{7087387} and used for
$2\omega \leq \frac{2\zeta_{max}}{\sqrt{3}}$, with the
procedure MinCost$\_$RP and used for
$\frac{2\zeta_{max}}{\sqrt{3}} < 2\omega <  2\zeta_{max}$.

\begin{algorithm} \caption{Opt$\_$MCLP($\omega$, $\ell$)} \label{alg:opt}
\begin{algorithmic}[1]

\State $S \gets \emptyset$; $minValue \gets \infty$

\If {$\omega \leq \frac{\zeta_{max}}{\sqrt{3}}$}

    \State Use ComputeMinCost function in \cite{7087387} to get the total cost $minValue$ and the placement sequence $S$ by given $\omega$ and $\ell$

\Else

    \State $(minValue, S) \gets$ MinCost$\_$RP ($\omega$, $\ell$)

\EndIf

\State \Return $(minValue, S)$


\end{algorithmic}
\end{algorithm}

\section{Performance Evaluation}\label{section:Simulation}

\begin{figure}[h]
\vspace{-20pt}
\center
\subfigure[]{\includegraphics[width=4cm]{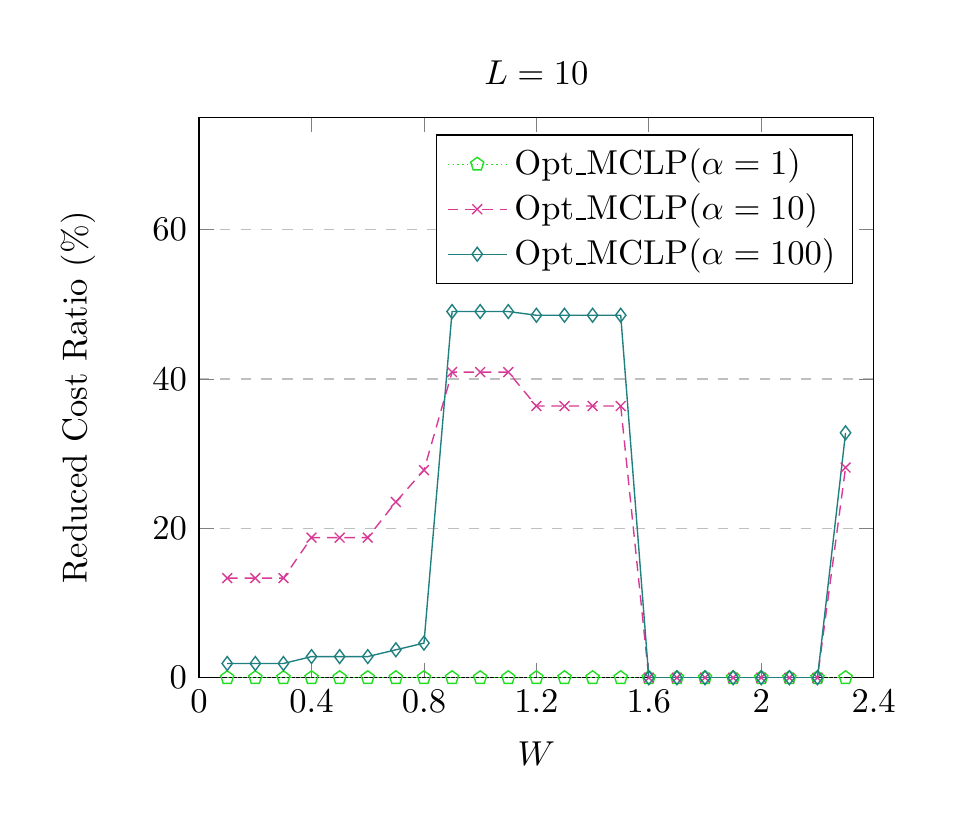}\label{Fig:sim1_ratio_L10}}
\subfigure[]{\includegraphics[width=4cm]{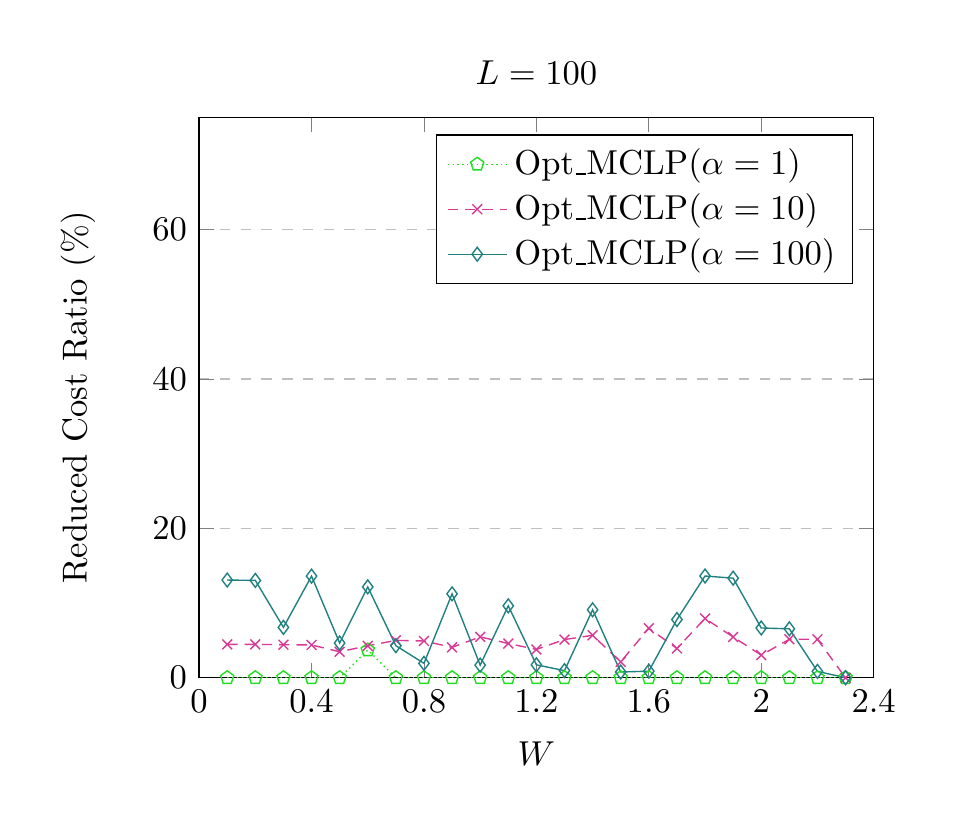}\label{Fig:sim1_ratio_L100}}
\caption{The reduced cost ratio when $\zeta_{max} = 2$, $\mu = 0.1$, $W$ ranges from $0.1$ to $2.3$, and $\alpha$ ranges from $1$ to $100$.
The values of $L$ are $10$ in (a) and $100$ in (b), respectively.}
\label{Fig:sim1_ratio}
\end{figure}

\begin{figure}[h]
\center
\subfigure[]{\includegraphics[width=4cm]{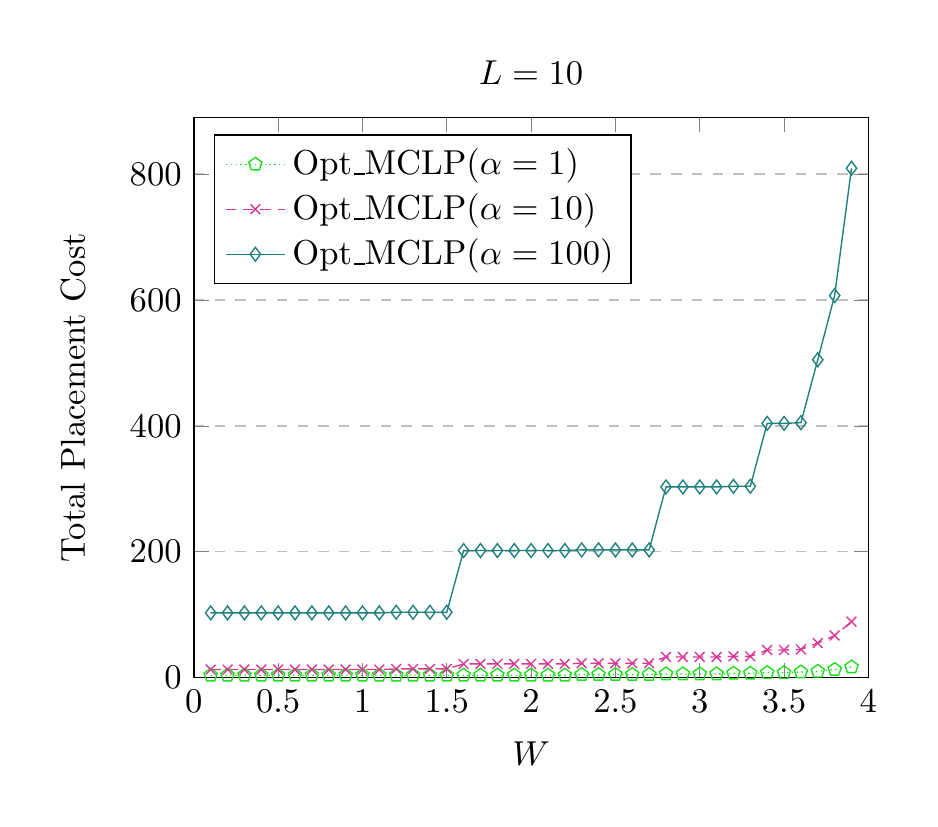}\label{Fig:sim2_cost_L10}}
\subfigure[]{\includegraphics[width=4cm]{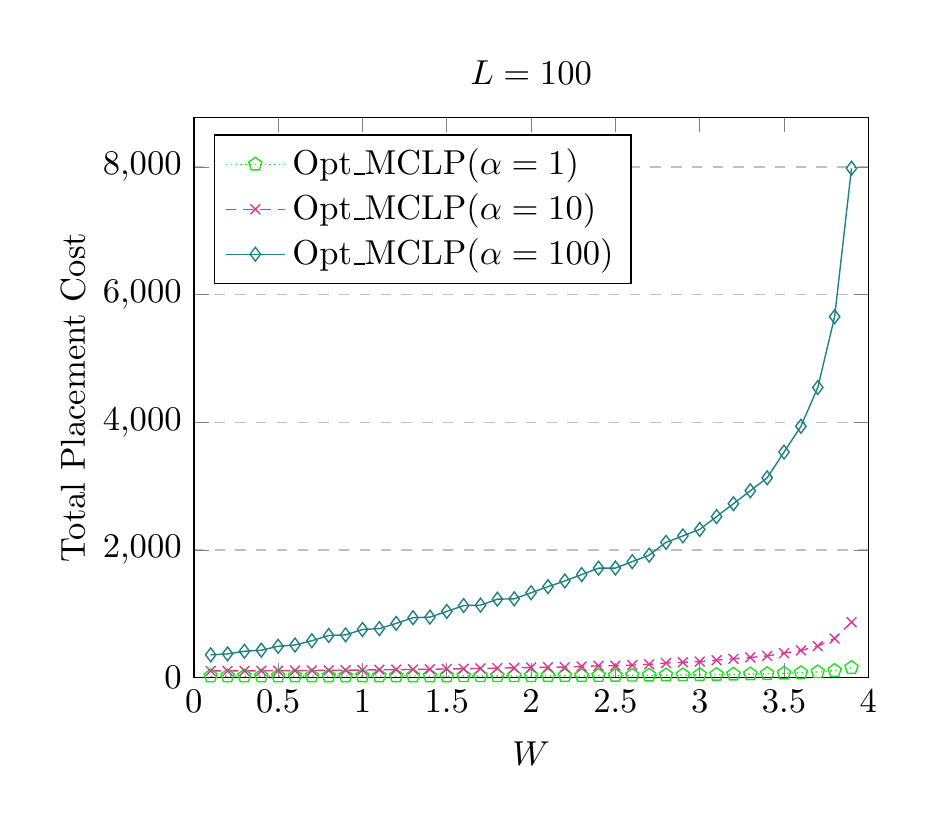}\label{Fig:sim2_cost_L100}}
\caption{The total placement cost when $\zeta_{max} = 2$, $\mu = 0.1$, $W$ ranges from $0.1$ to $3.9$, and $\alpha$ ranges from $1$ to $100$.
The values of $L$ are $10$ in (a) and $100$ in (b), respectively.}

\label{Fig:sim2_cost}
\end{figure}

Here, simulations developed by C++ were used to
evaluate the performance of the Opt$\_$MCLP for the MCLP problem.
In the simulations, $\mu$ and $c_r$ were set to $0.1$ and $1$, respectively.
In addition, $W$, $L$, $\zeta_{max}$, and $\alpha$ were set from $0.1$ to $10$, from $10$ to $100$, from $1$ to $4$, and from $1$ to $100$, respectively,
where $\alpha$ denoted $\frac{c_t}{c_r}$.
Let the reduced cost ratio of $B$ compared with $A$ be $\frac{c_A-c_B}{c_A}$,
where $A$ and $B$ were solutions of the MCLP problem and $c_A$ (or, $c_B$) denoted
the total placement cost required by $A$ (or, $B$).
For the MCLP problem,
because transmitters and receivers had to be deployed on a line to form a linear barrier,
the Algorithm 2 in \cite{7087387} with one linear barrier was compared
in terms of the reduced cost ratio in $\S$\ref{subsection:MCLP_Problem}.
Because the Algorithm 2 proposed in \cite{7087387}
can be only used for covering a rectangle area with width
less than or equal to $\frac{2\zeta_{max}}{\sqrt{3}}$,
$W$ was set from $0.1$ to $2.3$ when $\zeta_{max} = 2$.
In addition, due to the fact that
the width of the $Area$ had to be less than $2\zeta_{max}$ in the MCLP problem,
we also evaluated the total placement cost of the Opt$\_$MCLP in $\S$\ref{subsection:MCLP_Problem}
when $\zeta_{max} = 2$ and $W$ ranged from $0.1$ to $3.9$.

\subsection{The MCLP Problem}\label{subsection:MCLP_Problem}

Let the reduced cost ratio of the Opt$\_$MCLP($\alpha=k$)
denote the reduced cost ratio of the Opt$\_$MCLP compared with the Algorithm 2 in \cite{7087387} having one linear barrier when $\alpha=k$.
Fig. \ref{Fig:sim1_ratio_L10} and Fig. \ref{Fig:sim1_ratio_L100}
show the reduced cost ratios of the Opt$\_$MCLP($\alpha=1$ and Opt$\_$MCLP($\alpha=100$), respectively,
when the values of $L$ are equal to $10$, $50$, and $100$, respectively.
It is clear that the Opt$\_$MCLP provides better performance than the Algorithm 2 proposed in \cite{7087387} with one linear barrier.
In addition, the Opt$\_$MCLP reduces the placement cost by up to $49\%$ in comparison with
the other method.

Let the total placement cost of the Opt$\_$MCLP($\alpha=k$)
denote the total placement cost of the Opt$\_$MCLP when $\alpha=k$.
Fig. \ref{Fig:sim2_cost_L10} and Fig. \ref{Fig:sim2_cost_L100}
show the total placement costs of the Opt$\_$MCLP($\alpha=1$) and Opt$\_$MCLP($\alpha=100$), respectively,
when the values of $L$ are equal to $10$, $50$, and $100$, respectively.
It is clear that in Fig. \ref{Fig:sim2_cost_L10} or Fig. \ref{Fig:sim2_cost_L100},
the total placement cost of the Opt$\_$MCLP increases with the increasing of $\alpha$.
This is because the cost of a transmitter increases when $\alpha$ increases.
In addition, it is also clear that
the greater the value of $W$, the higher the total placement cost
required by the Opt$\_$MCLP($\alpha=1$), Opt$\_$MCLP($\alpha=10$), or Opt$\_$MCLP($\alpha=100$).
This is due to the fact that transmitters and receivers are closely deployed in order to cover the $Area$ with greater width,
and therefore, more transmitters and receivers are required.
Moreover,
when the value of $L$ increases from $10$ to $100$,
the total placement cost of the Opt$\_$MCLP($\alpha=1$), Opt$\_$MCLP($\alpha=10$), or Opt$\_$MCLP($\alpha=100$)
increases because more area has to be covered.

\section{Conclusion}\label{section:Conclusion}
The sensing model of bistatic radars (Cassini oval sensing model \cite{7087387}) is in fact very complex and
the shape of sensing areas are, therefore, varied according to the variation of distance between the transmitter(s) and the receiver(s) in the bistatic radar system.
In this paper, we study the MCLP problem for constructing a belt barrier with minimum placement cost.
For the MCLP problem with $0 \leq 2\omega \leq \frac{2\zeta_{max}}{\sqrt{3}}$,
a function, termed the MinCost$\_$RP, is proposed to find an optimal solution. In addition, for $\frac{2\zeta_{max}}{\sqrt{3}} < 2\omega <  2\zeta_{max}$, an optimal solution, termed the Opt$\_$MCLP, is proposed.
Theoretical analysis is also provided for proving the optimality of the Opt$\_$MCLP.
Simulation results show that the Opt$\_$MCLP has a significantly lower placement cost than the existing solution.

\normalem

\bibliographystyle{IEEEtran}
\bibliography{my}

\end{document}